\documentclass[preprint,superscriptaddress,
showpacs,preprintnumbers,amsmath,amssymb]{revtex4}

\usepackage{amssymb,amsmath,amsthm}
\usepackage{epsfig}

\newcommand*{\cB}{\mathcal{B}}
\newcommand*{\cC}{\mathcal{C}}

\newcommand*{\cE}{\mathcal{E}}
\newcommand*{\cF}{\mathcal{F}}
\newcommand*{\cH}{\mathcal{H}}

\newcommand*{\cK}{\mathcal{K}}

\newcommand*{\cN}{\mathcal{N}}
\newcommand*{\cP}{\mathcal{P}}
\newcommand*{\cO}{\mathcal{O}}

\newcommand*{\cS}{\mathcal{S}}

\newcommand*{\cU}{\mathcal{U}}

\newcommand*{\bbR}{\mathbb{R}}
\newcommand*{\bbC}{\mathbb{C}}
\newcommand*{\bbE}{\mathbb{E}}
\newcommand*{\bbP}{\mathbb{P}}

\newcommand*{\bbI}{\mathbb{I}}

\newcommand*{\sph}{\cS^{2d-1}}
\newcommand{\eps}{\varepsilon}
\newcommand{\vcd}{\emph{VC-dim}(\cF)}
\newcommand{\vcc}{\emph{VC-dim}(\cC)}
\newcommand{\tr}{\textnormal{tr}}
\newcommand{\ket}[1]{| #1 \rangle}

\newcommand{\inn}[2]{\langle #1 | #2 \rangle}

\theoremstyle{plain}
\newtheorem{lemma}{Lemma}
\newtheorem{theorem}[lemma]{Theorem}

\theoremstyle{definition}
\newtheorem{definition}{Definition}

\begin{document}
\title{Discretization of quantum pure states and local random unitary channel}

\author{Dong Pyo Chi}
\affiliation{
 Department of Mathematical Sciences,
 Seoul National University, Seoul 151-742, Korea
}
\author{Kabgyun Jeong}
\affiliation{
 Nano Systems Institute (NSI-NCRC),
 Seoul National University, Seoul 151-742, Korea
}

\date{\today}

\begin{abstract}
We show that a quantum channel $\cN$ constructed by averaging over
$\cO(\log d/\eps^2)$ randomly chosen unitaries gives a local
$\eps$-randomizing map with non-negative probability. The idea comes
from a small $\eps$-net construction on the higher dimensional
unit sphere or quantum pure states.
By exploiting the net, we analyze the concentrative phenomenon
of an output reduced density matrix of the channel,
and this analysis imply that there exists
a local random unitary channel, with relatively small unitaries, generically.
\end{abstract}

\maketitle

\section{Introduction}

The probabilistic existence of an $\eps$-randomizing map
or random unitary channel, with small cardinality of unitaries,
has several important implications. The channel can be used to
construct almost perfectly secure encryption protocols~\cite{AMTW00, HLSW04}
and give an intuition such as counterexample to additivity conjecture for
the classical capacity of quantum channel~\cite{HW08, CHLMW08, Has09}.

In the paper, we prove that there exists a quantum channel
consisting of unitary matrices
with relatively small cardinality $\cO(\log d/\eps^2)$
which is also $\eps$-randomizing.
This construction deeply relies on the mathematical fact
known as general $\eps$-net theorem,
in special, we consider a \emph{higher dimensional unit sphere}
corresponding to $d$($\gg 1$) dimensional quantum pure states.

For convenience, we use the following notations throughout the paper.
A \emph {state} can be pure or mixed state on the Hilbert spaces.
Especially, a density matrix of the pure state $\ket\varphi$
will be denoted as $\varphi$, when it is without confusing a mixed state.
If $\varphi_{AB}$ is a composite quantum state of
$\bbC^{d_Ad_B}~(\equiv\bbC^{d_A}\otimes\bbC^{d_B})$, or simply $A\otimes B$,
the reduced state on $A$ can be referred to $\varphi_A$.
Given Hilbert space $\cH(\bbC^d)$, $\cB(\bbC^d)$ denotes
the algebra of complex $d\times d$ matrices,
$\cU(d)$ be the unitary group on the space,
and $\bbI$ is $d\times d$ identity matrix.
The notation $\bbP[X]$ and $\bbE[X]$ denote the probability and
the expectation value
of a given random variable $X$, respectively. Finally,
any functions $\log$ and $\exp$ are always taken base 2.

\subsection{Local random unitary channel}
\label{sec:lruc}

A \emph{quantum channel} is a completely positive trace-preserving~(CPT) map
$\cN :\cB(\bbC^{d_A})\to\cB(\bbC^{d_B})$.
Given CPT map $\cN$, it is known that there is a complementary or conjugate channel
$\cN^C:\cB(\bbC^{d_A})\to\cB(\bbC^{d_E})$\cite{Hol05, KMNR07}.
For any input $\varphi_A$, these two channels
$\cN(\varphi_A)$ and $\cN^C(\varphi_A)$ are related by
\begin{equation}
\cN(\varphi_A)=\tr_{E} V\varphi_A V^\dagger, ~~~~\cN^C(\varphi_A)=\tr_{B} V\varphi_A V^\dagger
\end{equation}
where $V:\bbC^{d_A}\to\bbC^{d_Bd_E}$ is a unitary embedding.

We now introduce general notion of the random unitary channel.
For any input quantum states $\varphi_A$,
the \emph{random unitary channel}~(RUC)
$\cN :\cB(\bbC^{d_A})\to\cB(\bbC^{d_B})$ can be described by
\begin{equation*}
  \cN(\varphi_A)=\sum_{i=1}^{d_E} \omega_i U_i \varphi_A U_i^\dagger,
\end{equation*}
where the weights, $\omega_1, \ldots, \omega_{d_E}$, are positive values
such that $\sum_i^{d_E}\omega_i=1$ and the operators $U_1, \ldots, U_{d_E}$ are
some unitary $d_A\times d_A$ matrices. When the positive weights are
all equal to $1/{d_E}$, RUC will be written as
$\cN(\varphi_A)=\sum_{i=1}^{d_E} \frac{1}{d_E} U_i \varphi_A U_i^\dagger$.
In this place, a map $\cN$ is called $\eps$-\emph{randomizing} if, for all input $\varphi_A$,
\begin{equation*}
  \left\|\cN(\varphi_A)-\frac{\bbI}{d_B}\right\|_\infty \leq \frac{\eps}{d_B},
\end{equation*}
where $\bbI$ be $d_B\times d_B$ identity matrix and $\eps$ is a
small positive number upper bounded by 1. The operator norm
$\|\rho\|_\infty$ of any $\rho$ can be taken to be the square root of the
largest eigenvalue of $\rho^\dagger\rho$. That is, we call a CPT map $\cN$ as
random unitary channel, if, for all inputs $\varphi_A$, the map
$\cN(\varphi_A)$ is $\eps$-randomizing.
In sense of conjugate channel,
$\left\|\cN^C(\varphi_A)-\bbI/{d_E}\right\|_\infty \leq \eps/{d_E}$
also can be defined as an $\eps$-randomizing map.

For future works, we need some extended notions
concerning to the random unitary channel and $\eps$-randomizing.
\begin{definition}\label{def:lruc}
    Assume that $\cN :\cB(\bbC^{d_A})\to\cB(\bbC^{d_B})$ is a CPT map.
    For all input states $\varphi_A$, if
    \begin{equation}
    \left\|\cN(\varphi_A)\right\|_\infty-\frac{1}{d_B} \leq \frac{\eps}{d_B}
    ~~~~\mathrm{and}~~~~
    \left\|\cN^C(\varphi_A)\right\|_\infty-\frac{1}{d_E} \leq \frac{\eps}{d_E},
    \end{equation}
    then $\cN$ is called a \emph{local} $\eps$-\emph{randomizing}.
\end{definition}
For sufficiently large $d_A\gg d_E$, generically
the output states of the channel are distributed as
\begin{equation} \label{eq:chnloutput}
\cN(\varphi_A)\simeq
\begin{pmatrix}
\frac{1}{d_B} &    0   &   {}           &  \cdots & {}     &  0\\
0             & \ddots &   {}           &  {}     & {}     & {}\\
{}            &   {}   &  \frac{1}{d_B} &  {}     & {}     & {}\\
\vdots        &   {}   &   {}           &  0      & {}     & \vdots\\
{}            &   {}   &   {}           &  {}     & \ddots & {}\\
0             &   {}   &   \cdots       &  {}     & {}     & 0
\end{pmatrix}
~~~~\mathrm{and}~~~~
\cN^C(\varphi_A)\simeq
\begin{pmatrix}
\frac{1}{d_E} &   {}   &   0           \\
{}            & \ddots &   {}           \\
0             &   {}   &  \frac{1}{d_E}
\end{pmatrix},
\end{equation}
as well as $\frac{1}{d_B}$ is almost equal to $\frac{1}{d_E}$.
Naturally, one can take the definition of \emph{local random unitary channel}
from the local $\eps$-randomizing.
In this case, note that any output states of the channel
$\cN(\varphi_A)$ do not need to having full rank;
it should be a partially randomized states, for example, $\cN(\varphi_A)$
in Eq.~(\ref{eq:chnloutput}).

Recently, it was shown that, for all $\eps\in(0,1]$, $\eps$-randomizing maps exist
in sufficiently large dimension $d_A$ such that $d_E$ can be taken to be
$\cO(d_A\log d_A/\eps^2)$ in~\cite{HLSW04} and $\cO(d_A/\eps^2)$ in~\cite{Aub09}
for the Haar distributed $U_i$, respectively. The proof of the theorems is based on
a large deviation technique
and discretization of quantum pure states via
$\eps$-net construction~\cite{BHLSW05, HLSW04}.
In this paper, we show that there is a small set of local random unitary channel
with cardinality $\cO(\log d_A/\eps^2)$ only, which is local $\eps$-randomizing.

\begin{theorem}\label{thm:slruc}
Let all $\eps\in(0, \frac{1}{2}]$ and $d_A$ is sufficiently large.
Let $\{U_i: 1\leq i \leq d_E\}$ with $d_E=\cO(\log d_A/\eps^2)$
be i.i.d. random unitaries
distributed according to the Haar measure on $\cU(d_A)$. Then the quantum channel
$\cN(\varphi_A)=\frac{1}{d_E}\sum_i^{d_E}U_i\varphi_A U_i^\dagger$ is a
local $\eps$-randomizing map with non-negative probability.
\end{theorem}

This construction deeply relies on the fact of discrete geometry known as
general $\eps$-net theorem~\cite{Mat02}, especially we consider an object
such as higher dimensional unit sphere $\cS^{2d_A-1}$,
it is corresponding to $d_A$ dimensional quantum pure states.
By some properties of an entangled random subspace,
the local $\eps$-randomizing map has its quantum informational meaning.

\subsection{Entangled random states}
\label{sec:entransub}

As mentioned above, let's consider the unitary embedding
$V:\bbC^{d_A}\to\bbC^{d_Bd_E}$.
Let's define $\varphi_{BE}\in\bbC^{d_Bd_E}$ be a higher dimensional
bipartite state.
Especially, a set $\cP(\bbC^{d_Bd_E})$ denotes
the set of all bipartite pure states lying on $\bbC^{d_Bd_E}$.
For the pure state $\ket{\varphi}_{BE}$,
it is known that there exists a unique, and unitarily
invariant, uniform distribution $\mu_h$, which is given by the Haar measure
on the unitary group $\cU(d_Bd_E)$.
Also there is a uniform measure for its subspaces $\bbC^{d_S}\subset \bbC^{d_Bd_E}$
that is unitarily invariant.
A \emph{random pure state} $\ket{\varphi}_{BE}$ is defined as $\ket{\varphi}_{BE}$ is any
random variable drawn by $\mu_h$ on $\cP(\bbC^{d_Bd_E})$.
Similarly a \emph{random pure sub-states} 
$\ket{\varphi}_S$ with dimension $s$ is any pure states induced
by unitarily invariant measure
on $\cP(\bbC^{s})\subset \bbC^{d_Bd_E}~$\cite{HLW06}.

Page's conjecture~\cite{LP88, Pag93} states that the average von Neumann entropy of
$\tr_E(\varphi_{BE})$ has almost all maximum value, that is,
any random pure states are near-maximally entangled state
on a bipartite higher dimensional space~\cite{FK94, San95, Sen96}.
For $\ket{\varphi}_S\in \cP(\bbC^s)$,
it is also near-maximally entangled state~\cite{HLW06}:

\begin{theorem}[Entangled Subspaces]
\label{thm:hlw06}
Let $\bbC^{d_Bd_E}$ be a bipartite system with dimension $d_Bd_E$
$(d_B\geq d_E\geq 3)$ and $0<\alpha<\log d_B$.
Then, with high probability, there exists a subspace $\cP(\bbC^s)\subset \bbC^{d_Bd_E}$
of dimension
$s=\cO\left(d_Bd_E\left(\frac{\alpha}{\log d_B}\right)^{5/2}\right)$
s.t. all states $\ket\varphi_{S}\in \bbC^s$ have entanglement at least
\begin{equation}
 E(\varphi_{S})=S(\varphi_B)\geq\log d_B-\alpha-\frac{1}{\ln2}\frac{d_B}{d_E},
\end{equation}
where $S(\varphi_B)$ is von Neumann entropy of $\varphi_B$.
\end{theorem}

Note that the dimension $s$ of $S$ is surely less than
the total dimension $d_Bd_E$ of $\bbC^{d_Bd_E}$ as well as
$\varphi_S\in \cP(\bbC^s)$ is almost entangled state.
Recall that all bipartite quantum pure state can be written
in \emph{Schmidt decomposition} form,
$\ket\varphi_{BE}=\sum_{i=1}^{\min\{d_B,d_E\}}
\sqrt{\lambda_i}\ket{e_i}_B\ket{f_i}_E$, where
$_B\inn{e_i}{e_j}_B=\delta_{ij}= {}_E\inn{f_i}{f_j}_E$
and $\sqrt{\lambda_i}$ is the Schmidt coefficients,
furthermore, if $\lambda_i$ are all equal, then
it is maximally entangled state.

As mentioned above, a local random unitary channel may induce
a partially randomized quantum state $\varphi'_E$ with some low-rank
less than $\varphi_B$ (having full-rank) of RUC's output.
Imagine a purification of near-maximally mixed state $\varphi_B$,
resulting state $\ket{\varphi}_{BE}$ will be a maximally entangled state
on $\bbC^{d_Bd_E}$.
For the same reason, any purification of $\varphi'_E$ also can be
considered as a maximally entangled state, $\ket{\varphi}_S$,
on the random subspace $S$ with the dimension $s$.

\section{Small $\eps$-net on Unit Sphere}
\label{sec:seps}

In this section, we define several important notions and
investigate some of their mathematical facts concerned to
an $\eps$-net on the unit sphere. Especially $\sph$ denote
a higher dimensional unit sphere on $\bbR^{2d}$, which is
generally corresponding to all quantum pure states on $\bbC^d$.
Now we show that there exists a small $\eps$-net $N$ for $\sph$
with cardinality $|N|=\cO(d\log(\eps^{-1})/\eps)$.

Let $(X, \cF)$ be a $\mu$-measurable set system and $\cF\subseteq X$,
here $\mu$ be a natural probability measure on $X$. For every
$\eps\in[0, 1]$, an $N\subseteq X$ is called an $\eps$-\emph{net}
for the system $X$ with respect to $\mu$ if $N\cap F_i\neq\emptyset$
for all $F_i\in\cF$ with $\mu(F_i)\geq\eps$~\cite{Mat02}.
To describe the $\eps$-net above, we need to a new parameter $\vcd$ of $X$,
which is called \emph{Vapnik-Chervonenkis} or just simply VC dimension of $\cF$.

\begin{definition}
  Let $\cF$ be a subset on $X$. Assume that another $A\subseteq X$ is
  \emph{shattered} by $\cF$ if $\cF|_A=2^A$, i.e., the restriction of $\cF$ on $A$
  gives a power set of $A$. Then the VC dimension of $\cF$ is defined:
  \begin{equation}
  \vcd=\sup_{A\subseteq X}\left\{|A|:\cF|_A=2^A\right\}.
  \end{equation}
\end{definition}

The restriction of $\cF$ on $A$ is defined by $\cF|_A=\{F_i\cap A: F_i\in\cF\}$.
It is well known that a system $\cF$ of all half-planes in the plane
$\bbR^2$ have $\vcd=3$ in~\cite{Mat02}. If an $m$-point subset $A$ lies in $X$, then
the \emph{shatter function} of $\cF$ is defined by
\begin{equation*}
\sigma_\cF(m)=\max_{A\subseteq X,~|A|=m}\left|\cF|_A\right|.
\end{equation*}
In other words, $\sigma_\cF(m)$ is the maximum possible value of distinct
intersections of the sets of $\cF$ with $A\subseteq X$.
For $\vcd\leq d $, the shatter function satisfies
that $\sigma_\cF(m)\leq\sum_{j=1}^d{}_mC_j$.
(This bound is known as \emph{shatter function lemma}~\cite{Mat02}.)

In this paper, we substitute $X$ and $\cF$ to $\sph$ and a cap, $\cC$,
respectively. Formally,
$\sph := \{\ket x\in\bbC^d:\|\ket x\|_2=1\}$.
Let's consider a uniform probability measure $\mu$ on $\sph$.
For any measurable subset $S\subset\sph$,
\begin{equation*}
  \mu(S)
  =\frac{\mathrm{vol}(S)}{\mathrm{vol}(\sph)}=\mathrm{vol}(S),
\end{equation*}
where the second equality follows from $\mu(\sph)=1$ by definition.
A \emph{cap} on $\sph$ is defined:
\begin{equation} \label{eq:cap}
  \cC=\sph\cap\{\ket x:\inn{u}{x}\geq1-h\}
\end{equation}
for some unit vector $\ket u\in\sph$
(exactly, $\ket u$ is the \emph{center} of $\cC$),
we refer to $h$ as the \emph{height} of the cap.
Note that $\cC$ can be considered as a (geodesic) convex set on $\sph$ with $\mu(\cC)>0$.
In such a cap, we know that, for all $h\leq\frac{1}{2}$,
the asymptotic radius and their ($2d-1$)-dimensional volume of $\cC$
are bounded by $\Theta(h^{1/2})$ and $\Theta(h^{(2d-1)/2})$ as $h\to 0$, respectively.
Next lemma states $\vcc$ on the higher dimensional unit sphere.

\begin{lemma}
The VC dimension of all closed cap $\cC$ on $\sph$ is equal to $2d+1$.
\end{lemma}

\begin{proof}
By Randon's lemma
(This lemma states that any ($d+1$)-point set on $\bbR^d$ can be
shattered by the system of all closed half-space.),
any set of $2d$ affinely independent points on the ($2d-1$)-dimensional
unit sphere can be shattered~(See e.g. Lemma 10.3.1 in~\cite{Mat02}), and then
$\vcc$ on $\sph$ is equivalent to $2d$. For every $h\in(0,1)$, all closed cap on $\sph$
allow a factor of the additional $1$ dimension.
\end{proof}

For example, the system of all closed half-space $\cF$ on $\cS^2$ has $\vcd=3$,
but the cap on a sphere enlarging $\vcc=4$. Note that $\cS^2$
corresponds to exactly $\cS^3$ in the above arguments; the difference comes from
the convenience $\bbC^d \cong \bbR^{2d}$ instead of $\bbC^d \cong \bbR^{2d-1}$.
Here we need an additional lemma for the proof of following theorem
~(Theorem~\ref{thm:epsNetSph}),
and see also details of the proof of Lemma 10.2.6 in \cite{Mat02}.

\begin{lemma} \label{lem:Chernoff-type}
Let $X=X_1 + \cdots + X_t$, where the $X_i$ are independent random variables,
\begin{eqnarray*} X_i=
\begin{cases}
  1 &         \mathrm{{\it with ~probability}}~\eps, \\
  0 &         \mathrm{{\it otherwise}}. \\
\end{cases} \end{eqnarray*}
Then $\bbP\left[ X\geq\frac{1}{2}t\eps \right]\geq\frac{1}{2}$, when $t\eps\geq8$.
\end{lemma}

For the higher dimensional unit sphere $\sph$, we can construct a small
$\eps$-net which may be almost optimal.

\begin{theorem}
\label{thm:epsNetSph}
Let $\mu$ be a uniform probability measure on $\sph$, $\cC\subset\sph$ be a cap
of $\mu$-measurable subsets with VC-dim$(\cC)\leq 2d+1$. If $d\geq 1$, and $\eps\leq\frac{1}{2}$,
then there exists an $\eps$-net $N$ for the set system $(\sph, \cC)$ w.r.t. $\mu$ of cardinality
\begin{equation} \label{eq:netcardi}
 |N|=\cO\left(d\frac{1}{\eps}\log\frac{1}{\eps}\right).
\end{equation}
\end{theorem}

\begin{proof}
  The proof is almost
  equivalent to the proof of Theorem 10.2.4 in~\cite{Mat02},
  on the other hand our proof has an essential difference by using the cap $\cC$,
  Eq.~(\ref{eq:cap}) above, on $\sph$.
  First of all, let's define three random samples ${\it \Sigma}_1,
  {\it \Sigma}_2$ and ${\it \Sigma}_3$.
  Assume that $t=\left\lceil Cd\frac{1}{\eps}\log\frac{1}{\eps}\right\rceil$,
  and ${\it{\Sigma}}_1$ be a random sample drawn from $t$
  independent random draw on $\sph$,
  where each elements satisfy the probability measure $\mu$.
  W.l.o.g., all $\cK_i\in\cC$ hold $\mu(\cK_i)\geq\eps$.
  By $t$ more independent random draw (of another purpose),
  we pick some random sample ${\it{\Sigma}}_2\subset \sph$ and
  fix an integer $k=t\eps/2$.
  Finally, let's define a fixed ${\it{\Sigma}}_3$, which is a random sample picked by
  $2t$ independent random draw from $\sph$ and fix a set $\cK^*\in\cC$.

  Now we consider two events $\cE_1$ and $\cE_2$.
  Let $\cE_1$ be the event so that the random sample ${\it{\Sigma}}_1$ fails to be
  an $\eps$-net, i.e., ${\it{\Sigma}}_1\cap\cK_i=\emptyset$ for all $\mu(\cK_i)\geq\eps$.
  Similarly, $\cE_2$ be the event such that there exists an $\cK_i\in\cC$
  with ${\it{\Sigma}}_1\cap \cK_i=\emptyset$ and $|{\it{\Sigma}}_2\cap \cK_i|\geq k$.
  Clearly $\cE_2$ needs $\cE_1$ plus something more condition,
  so $\bbP[\cE_2]\leq\bbP[\cE_1]$.
  We need another probabilistic condition such that $\bbP[\cE_2]\geq\bbP[\cE_1]/2$.
  Suppose that there is $\cK_i$ with ${\it{\Sigma}}_1\cap \cK_i=\emptyset$,
  and let's fix one of them $\cK^*$.
  Then $\bbP[\cE_2|{\it{\Sigma}}_1]\geq\bbP[|{\it{\Sigma}}_2\cap \cK^*|\geq k]
  \geq\frac{1}{2}$. The value of $|{\it{\Sigma}}_2\cap \cK^*|$ behaves like
  the random variable $X=X_1 + \cdots + X_t$.
  By using Lemma~\ref{lem:Chernoff-type},
  above second inequality holds.
  So $2\bbP[\cE_2|{\it{\Sigma}}_1]\geq\bbP[\cE_1|{\it{\Sigma}}_1]$
  for all ${\it{\Sigma}}_1$, and thus $2\bbP[\cE_2]\geq\bbP[\cE_1]$.

  Next, we must bound the distribution $\bbP[\cE_2]$.
  If we define a conditional probability
  $P_{\cK^*}=\bbP[{\it{\Sigma}}_1\cap \cK^*=\emptyset,
  |{\it{\Sigma}}_2\cap \cK|\geq k|{\it{\Sigma}}_3]$,
  then
  \begin{align*}
    P_{\cK^*}
    & \leq
    \bbP[{\it{\Sigma}}_1\cap \cK^*=\emptyset|{\it{\Sigma}}_3]
    =\frac{{}_{2t-k}C_t}{{}_{2ts}C_t} \leq \left(1-\frac{k}{2t}\right)^t \\
    & \leq
    e^{-(k/2t)t} 
    = e^{-(Cd\log(1/\eps))/4}
    = \eps^{Cd/4}.
  \end{align*}

  Finally, we exploit the assumption of the $\vcc$, which
  any set of $\cC$ have at most $\sum_{j=0}^{2d+1}{}_{2t}C_j$
  distinct intersections with ${\it{\Sigma}}_3$, via the shatter function lemma.
  For all fixed ${\it{\Sigma}}_3$,
  \begin{align*}
    \bbP[\cE_2|{\it{\Sigma}}_3]
    & \leq
    \left({}_{2t}C_0 + \cdots+{}_{2t}C_{2d+1}\right)\times \eps^{Cd/4}   \\
    & \leq
    \left(\frac{2te}{2d+1}\right)^{2d+1}\times\eps^{Cd/4}
    = \left(\frac{2te}{2d+1}\right)^{2d+1}\times\left(\eps^{C'/4}\right)^{2d+1}   \\
    & = \left(2e(1/\eps)\log(1/\eps)\times \eps^{C'/4}\right)^{2d+1} < \frac{1}{2},
  \end{align*}
  if $d\geq1$, $\eps\leq1/2$ and some constant $C'$ is sufficiently large.
  So $\bbP[{\it{\Sigma}}_1]\leq 2\bbP[{\it{\Sigma}}_2]<1$,
  which completes the proof.
\end{proof}

If we define $d=d_B$, and for all $\eps\leq 1/2$, then there exists an $\eps$-net
$N$ for $\cS^{2d_B-1}$ with cardinality $|N|=\cO\left(d_B\frac{1}{\eps}\log\frac{1}{\eps}\right)$.
For the proof of Theorem~\ref{thm:slruc}, not only the above $\eps$-net construction
but also the following Lemma~\ref{lem:LDE} of a large deviation estimate are crucial,
and see details of the proof in~\cite{HLSW04, BHLSW05}. Note that in their proof
they use the equal dimension of input and output, i.e., $\cN:\cB(\bbC^{d_A})\to\cB(\bbC^{d_B})$
s.t. $d_A=d_B=d$.

\begin{lemma}
\label{lem:LDE}
Let $\varphi_A$ be a pure state, and $\Pi$ be a rank $p$ projector.
Let $\{U_i:1\leq i\leq d_E\}$ be a sequence of $\cU(d_A)$-valued i.i.d.
random variable, distributed according to  Haar measure. Then, for all
$\eps\in(0,1)$,
\begin{equation}
 \bbP\left[\left|\frac{1}{d_E}\sum_{i=1}^{d_E}\tr\left(U_i\varphi_A U_i^\dagger\right)
 -\frac{p}{d_B}\right|\geq\frac{\eps p}{d_B}\right]\leq2e^{-d_Ep\frac{\eps^2}{6\ln2}}.
\end{equation}
\end{lemma}

Unfortunately, the lemma directly cannot be applied to constructing
the local random unitary channel, because the operator norm concern to
a different output parameters.
On this account,
let's consider a concentrated phenomenon of the output reduced density matrices.

\section{Concentration of Reduced States}
\label{sec:concentr}

We have already mentioned that a random pure state
as well as its random pure sub-states are
almost surely maximally entangled in Section~\ref{sec:entransub}.
In special we take into account a concentration of reduced density matrices
of $\cP(S)\subset\bbC^{d_Bd_E}$, and improve the Theorem~\ref{thm:hlw06}
by using the $\eps$-net theorem (see Theorem~\ref{thm:epsNetSph}) and
large deviation technique~(Lemma\ref{lem:LDE}).

Recall the definition of local $\eps$-randomizing of a channel $\cN$:
$\|\cN(\varphi_A)\|_\infty-\frac{1}{d_B}\leq\frac{\eps}{d_B}$. Note that
the image of $\varphi_A$ under the unitary embedding $V:\bbC^{d_A}\to\bbC^{d_Bd_E}$
can be considered as a subspace $S$ of dimension $s$ in $\bbC^{d_Bd_E}$,
and it is highly entangled.

\begin{lemma} \label{lem:concRDM}
Let $\ket{\varphi}_{S}$ be a random pure state on $\bbC^{d_Bd_E}$,
and $\eps\in(0,1]$. Then
\begin{equation} \label{eq:concRDM}
\bbP\left[\|\cN(\varphi_A)\|_\infty-\frac{1}{d_B}\geq\frac{\eps}{d_B}\right]
\leq\left(C\frac{d_B^2}{\eps}\log\frac{d_B}{\eps}\right)e^{-\frac{d_E\eps^2}{14\ln2}}.
\end{equation}
\end{lemma}

\begin{proof}
By using the Crem$\acute{e}$r's rule~\cite{DZ93, HLSW04} and
for a squared Gaussian random random variable~\cite{BHLSW05},
we can obtain the following bound:
\begin{equation*}
\bbP\left[\frac{1}{d_E}\sum_{i=1}^{d_E}X_i\geq(1+\eps)\sigma^2\right]
\leq e^{-d_E\frac{\eps-\ln(1+\eps)}{2\ln2}}\leq e^{-\frac{(d_E) \eps^2}{14\ln2}},
\end{equation*}
where $\{X_i\}$ are some real-valued i.i.d. random variables
and $\sigma$ denotes a standard deviation of the distribution.
Let's substitute the parameters from
$\sigma^2$ and  $X_i$ to $\frac{1}{d_B}$ and
$\tr(\varphi_{B}\tr_E(U_i\psi_{BE} U_i^\dagger))$, respectively.
Here, $\psi_{BE}$ is a random pure state on $\cP(\bbC^{d_Bd_E})$,
but $\varphi_{B}$ in $\cP(\bbC^{d_B})$.
Then we obtain a new Crem$\acute{e}$r's bound such that
\begin{equation}
\bbP\left[\tr(\varphi_{B}\tr_E(U\psi_{BE} U^\dagger))-\frac{1}{d_B}\geq
\frac{\eps}{d_B}\right]\leq e^{-\frac{(d_E)\eps^2}{14\ln2}}.
\end{equation}
Let's denote $\varphi_B$ equal to $\varphi$ on the sphere,
and $\tilde{\varphi}_B$ just to $\tilde{\varphi}$ on the net, for shortly.
By exploiting the definition of conjugate channel and operator norm
induced by some pure state, we obtain a relation that
\begin{eqnarray}
\|\cN(\varphi_A)\|_\infty &=&
\|\tr_E(U\psi_{BE} U^\dagger)\|_\infty=\sup_{\varphi\in B}
\tr(\varphi\tr_E(U\psi_{BE} U^\dagger)) \label{eq:vecopnorn} \\
&=& \sup_{\varphi\in B}[\tr(\varphi\tr_E(U\psi_{BE} U^\dagger))
-\tr(\tilde{\varphi}\tr_E(U\psi_{BE} U^\dagger))]
+\sup_{\varphi\in B}\tr(\tilde{\varphi}\tr_E(U\psi_{BE} U^\dagger)) \nonumber\\
&=&\sup_{\tilde{\varphi}\in N_B}\tr(\varphi-\tilde{\varphi})\tr_E(U\psi_{BE} U^\dagger)
+\sup_{\tilde{\varphi}\in N_B}\tr(\tilde{\varphi}\tr_E(U\psi_{BE} U^\dagger)).
\end{eqnarray}

Above Eq.~(\ref{eq:vecopnorn}) followed by the definition of induced operator norm,
and the supremum in the last equality run over all points on the net $N_B$.
Now we fix $\frac{\eps}{2d_B}$-net $N_B$ for the system $B$,
then $|N_B|=\left(C\frac{d_B^2}{\eps}\log\frac{2d_B}{\eps}\right)$,
$C$ be an universal constat. Furthermore we use the fact which
if $\|\ket\varphi-\ket{\tilde{\varphi}}\|_1\leq\eps$, then
$\tr(\varphi-\tilde{\varphi})\Pi\leq\frac{\eps}{2}$, where
$\ket\varphi$ and $\ket{\tilde{\varphi}}$ are points on the unit sphere
and on the net, respectively.
$\Pi$ is a projector such that $\Pi\in[0,\bbI]$~\cite{HHL04}.
Thus
\begin{eqnarray*}
\|\cN(\varphi_A)\|_\infty &=&
\sup_{\tilde{\varphi}\in N_B}\tr(\varphi-\tilde{\varphi})\tr_E(U\psi_{BE} U^\dagger)
+\sup_{\tilde{\varphi}\in N_B}\tr(\tilde{\varphi}\tr_E(U\psi_{BE} U^\dagger)) \\
&\leq&
\sup_{\tilde{\varphi}\in N_B}\tr(\tilde{\varphi}\tr_E(U\psi_{BE} U^\dagger))
+ \frac{\eps}{4d_B}.
\end{eqnarray*}
Now, we use the union bound and Lemma~\ref{lem:LDE}.
Then, for some constant $C$,
\begin{eqnarray*}
\bbP\left[\|\cN(\varphi_A)\|_\infty-\frac{1}{d_B}\geq\frac{\eps}{d_B}\right]
&\leq& \bbP\left[\sup_{\tilde{\varphi}\in N_B}\tr(\tilde{\varphi}\tr_E(U\psi_{BE} U^\dagger))
-\frac{1}{d_B}\geq \frac{3\eps}{4d_B}\right] \\
&\leq& \left(C\frac{d_B^2}{\eps}\log\frac{d_B}{\eps}\right)e^{-\frac{d_E\eps^2}{14\ln2}}.
\end{eqnarray*}
\end{proof}

Let's briefly summarize the previous results for finishing proof of
our main theorem: Theorem~\ref{thm:slruc}.
Theorem~\ref{thm:epsNetSph} states that there exists an $\eps$-net
of cardinality $|N_B|=\cO\left(d_B\frac{1}{\eps}\log\frac{1}{\eps}\right)$
for a higher dimensional unit sphere $\cS^{2d_B-1}$ and its cap $\cC\subset\cS^{2d_B-1}$,
constrained by $\vcc\leq 2d_B+1$. By using the net, we have investigated that
the concentration of reduced density matrix which is almost maximally mixed state
with high probability. The equation (\ref{eq:concRDM}) in Lemma~\ref{lem:concRDM} describes
the concentration phenomenon of density matrix, furthermore a bound of the inequality imply
the proof of the main result.

\begin{proof}[Proof of Theorem~\ref{thm:slruc}]
Recall Eq.~(\ref{eq:concRDM}) in Lemma~\ref{lem:concRDM}.
We want to bound the right hand side of the inequality above by 1,
depending on $d_E$. Let's take
$d_E\geq C''\frac{\log d_E^2}{\eps^2}$ where $C''$ be a suitable constant, then
$\bbP\left[\|\cN(\varphi_A)\|_\infty-\frac{1}{d_B}\geq\frac{\eps}{d_B}\right]\leq 1$.
This bound straightforwardly means that our claim is true.

Finally, if we choose the dimension of input space of the channel $\cN$
equal to its output~($d_A=d_B$), then the proof of
Theorem~\ref{thm:slruc} will be completed;
For all $\eps\in(0, \frac{1}{2}]$ and $d_A$ is sufficiently large.
Let's $\{U_i\}_{i=1}^{d_E}$ with $d_E=\cO(\log d_A/\eps^2)$ be
an i.i.d. random unitary matrices
distributed according to the Haar measure on $\cU(d_A)$. Then the channel
$\cN(\varphi_A)=\frac{1}{d_E}\sum_i^{d_E}U_i\varphi_A U_i^\dagger$ is a
local $\eps$-randomizing map with positive probability.
\end{proof}

\section{Conclusions}
\label{Sec:Conclusions}

In conclusion, we have shown that a quantum channel such that local random unitary channel,
$\cN$, constructed by averaging over
$\cO(\log d/\eps^2)$ randomly chosen unitaries gives a local
$\eps$-randomizing map with positive probability. The whole idea begins
from not only the small $\eps$-net construction on the higher dimensional
unit sphere corresponding to the sufficiently larger dimensional quantum pure states,
but also the analyzing a phenomenon of concentration
of output reduced density matrix of the quantum channel. Generically, a higher dimensional bipartite
quantum pure states (or random pure states) are almost all maximally entangled,
and its entropy of the reduced density matrices are maximally mixed.

By using the result,
one could attempt to improving the bound on a communication resources for the private quantum channel,
quantum superdense coding and quantum data hiding etc.
Furthermore, one may investigate the minimal dimensions of the violation of additivity conjecture
for the minimum output entropy.

\section*{Acknowledgments}
This work was supported by Basic Science Research Program
through the National Research Foundation of Korea~(NRF) funded
by the Ministry of Education, Science and Technology~(Grant No. 2009-0072627).


\end{document}